\documentclass{eptcs} 
\usepackage[utf8]{inputenc} 
\usepackage{bbold} 
\newcommand\gt{\mathsf{gt}} 
\newcommand\Hom{\mathsf{Hom}} 

\usepackage{tikz-proofnets} 
\usetikzlibrary{positioning} 
\usetikzlibrary{decorations.markings} 
\tikzset{vertex/.append style={inner sep=1.5pt}} 
\usepackage{subfigure}

\usepackage{amsmath} 
\usepackage{amsthm} 
\usepackage{MnSymbol} 
\usepackage{stmaryrd} 
\usepackage{cmll} 
\newcommand\flag{\,^<\!\!|} 
\newcommand\paral\parr 
\newcommand\deux{\mathbb{2}} 
\newcommand\un{\mathbb{1}} 
\newcommand\unplusun{{\un\oplus\un}}
\newcommand\rr{\mathsf{r}}
\newcommand\id{\mathrm{Id}} 
\newcommand\dicog{\mathsf{dicog}} 
\newcommand\comsq[1]{\textsc{com.sq}($#1$)}

\usepackage{prooftree} 
\newtheorem{prop}{Proposition} 
\newtheorem{thm}[prop]{Theorem} 
\newtheorem{dfn}[prop]{Definition} 
\newcommand\bef{\mathbin{\triangleleft}} 
\newcommand\ts{\mathbin{\otimes}} 
\newcommand\pa{\mathbin{\bindnasrepma} } 
\newcommand\seq\vdash 
\renewcommand\ae{\textsc{ae}} 

\newcommand{\doi}[1]{\textsc{doi}: \href{http://dx.doi.org/#1}{\nolinkurl{#1}}}

\title{Flag: a Self-Dual Modality\\ for Non-Commutative Contraction 
and Duplication\\ in the Category of Coherence Spaces} 
\author{Christian Retor\'e \institute{LIRMM, Univ Montpellier \& CNRS } } 
\def\authorrunning{Christian Retor\'e} 
\def\titlerunning{Flag: a Self-Dual Modality  for Non-Commutative  Contraction and Duplication} 
\usepackage[utf8]{inputenc} 

\renewcommand\sp{\ifmmode{\scriptscriptstyle SP}\else{\sc sp}\fi} 
\newcommand\tablecoh[5]
{
\begin{array}{l||c|c|c|}
#1 #3 #2 & \sincoh & = & \scoh \\ \hline \hline 
\sincoh & \sincoh & \sincoh & #4 \\ \hline 
= & \sincoh & = & \scoh \\ \hline 
\scoh & #5  & \scoh & \scoh \\ \hline 
\end{array} 
} 

\newcommand\edge{\!\!\--\!\!} 
\begin{document} 

\maketitle 

%
%
%
%
%
%

\begin{flushright} \footnotesize \it 
This paper is dedicated to the memory of my parents,\\ both math teachers, both  deceased in 2021. 
\end{flushright} 

\begin{abstract}
After reminding what coherences spaces are and how they interpret linear logic, we define a modality ``flag" in the category of coherence spaces (or hypercoherences) with  two inverse linear (iso)morphisms:  ``duplication"  from (flag A) to 
((flag A) $\bef$ (flag A)) and ``contraction" in the opposite direction --- where $\bef$    is the self dual and non-commutative  connective  known as ``before" in pomset logic and known as ``seq(ential)" in the deep inference system (S)BV. 
In addition, as expected, the  coherence space A is a retract of its modal image (flag A). 
This suggests an intuitive interpretation of (flag A) as  ``repeatedly A" or as ``A at any instant" when  ``before"  is given a temporal interpretation. 
We hope the semantic construction of flag(A)  will help to design proof rules for ``flag" and we briefly discuss this at the end of the paper. 
\end{abstract}


\section{Presentation} 
Given the aftermath  \cite{gug2017lyon,Slavnov2019,nguyen2019proof,Retore2021Lambek}  of pomset logic \cite{Ret93,Ret97tlca} 
and of the related BV calculus of structures 
\cite{Gug99,GugStr01,GugStrass2011tcl}
we decided to publish this  ancient work  of ours \cite{Ret94mod}, which answers 
a question of Jean-Yves Girard in \cite{Gir91}. This semantic work might help to find the proper rules for a self dual modality as recently sketched by Alessio Guglielmi in \cite{gug2017lyon}. 

The structural rules of classical logic are responsible for the non-determinism of classical logic, and linear logic which carefully handles these rules is especially adequate for a constructive treatment of classical logic, as \cite{Gir91} shows.  Linear logic handles structural rules by the modalities (a.k.a. exponentials)
``?" and ``!". The modality ``?" allows contraction and weakening in positive position, and the modality ``!" in negative position. The formula $!A$ linearly implies $A$ while the formula $?A$ is linearly implied by $A$. 
In semantical words, this means we have the linear morphisms:

$$ 
\begin{array}{rclcrcl} 
!A&\multimap& (!A\otimes !A)&\qquad& ?A & \leftspoon  & ?A \paral ?A  \\
\downspoon &\sespoon &   &\qquad& \upspoon  & \nwspoon  & \\ 
1 & & A &\qquad& 1 &  & A \\ 
\end{array}
$$ 

The major difficulty when dealing with classical logic are the cross-cuts, appearing in the cut elimination theorem of Gentzen as a rule called MIX \cite{Gen34} . This rule is a generalised cut between several occurrences of $A$ and several occurrences of $\neg A$, i.e.   a cut between two formulas both coming from contractions. This is a major cause of non-determinism: an example can be found in \cite[Appendix B, Example 2, p. 294]{Gir91}. 
In linear logic, such a cut may not happen, since contraction only applies to $?A$ formulas, while their negation  is $!A^{\perp}$ which can not come from a contraction --- hence no such CUT is possible in linear logic.

Let us now quote the precise paragraph of \cite[p. 257]{Gir91} which initially motivated our work on a self-dual modality for contraction: 
\begin{quote} \it The obvious candidate for a classical semantics was of course coherence spaces which had already given birth to linear logic; the main reason for choosing them was the presence of the involutive linear negation. However the difficulty with classical logic is to accommodate structural rules (weakening and contraction); in linear logic, this is possible by considering
coherent spaces $?X$. But since classical logic allows contraction and weakening both on a formula and its negation, the solution seemed to require the linear negation of $?X$ to be of the form $?Y$, which is a nonsense (the negation of $?X$ is $!X^{\perp}$ which is by no means of isomorphic to a space $?Y$).
Attempts to find a self-dual variant $\S Y$ of $?Y$ (enjoying $(\S Y)^{\perp}=\S Y^{\perp}$ ) systematically failed. The semantical study of classical logic stumbled on this problem of self-duality for years.
\end{quote} 

Here as well we focus on coherence spaces because of their tight relation to linear logic \cite{Gir87,Tro92,Ret94rc,Ret97,Girard2011blindspot}. Once the modality is found in the category of coherence spaces, and shown to enjoy the expected properties, we briefly show that this modality also exists in the category of hypercoherences of  Thomas Ehrhard \cite{Ehr93}, and discuss the possibility to conceive syntactic rules for this modality. 

In our previous work on pomset logic \cite{Ret93,Ret97tlca,Retore2021Lambek}, we studied a self-dual and non-commutative  connective called ``before", written 
$\bef$ together with partially ordered multisets of formulae. This lead us to the modality $\flag$  (flag) to be described in this paper. Flag is an  endofunctor of the category of coherence spaces equipped with linear maps, flag is self-dual, and flag enjoys both \emph{left and right contraction} with respect to ``before"  as a pair of linear isomorphisms 
$\flag  A \leftspoon\!\rightspoon (\flag A\bef \flag  A)$
and $A$ is a retract of
$\flag  A$ . Fortunately, there is no weakening, which would yield unwanted morphisms like  $1 \multimap A$ and $A \multimap 1$ for all $A$.

So far, there is no syntax extending pomset logic to this modality, although Alessio Guglielmi sketched a possible syntax in the closely related calculus of structures in \cite{gug2017lyon}. In order to define a syntax, we should first study the basic steps of cut-elimination, in particular the contraction/contraction case and the commutative diagrams it requires, as Myriam Quatrini did in \cite{Qua95} for the logical calculus LC of \cite{Gir91}. We briefly discuss the guidelines for design of the rules for flag at the end of the paper.

\section{Denotational Semantics of Linear Logic with ``Before" Connective}
\label{coh}

This section is taken over from  \cite{Retore2021Lambek} and slightly adapted to the presentation of the flag modality. It is just a short presentation of coherences spaces and of the interpretation of linear logic in coherence spaces in order to provide a reasonably self-contained paper. 

\begin{dfn}[categorial interpretation] \label{catinterpret} 
Denotational semantics or categorical interpretation of a logic is the interpretation in a category, 
here called $\mathbf{C}$: 
\begin{itemize} 
\item formulas are interpreted as objects in $\mathbf{C}$
\begin{itemize} 
\item an object $\llbracket a \rrbracket$ is arbitrarily associated with any atomic formula $a$ (or propositional variable) 
\item an $n$-ary logical connective (usually $n=2$ or $n=1$)  $T[A_1,\ldots,A_n]$ is interpreted as the contruction of  $\llbracket T[A_1,\ldots,A_n] \rrbracket$ from the $n$ objects $\llbracket A_1\rrbracket$ \ldots $\llbracket A_n\rrbracket$.
\end{itemize} 
\item proofs are interpreted as morphisms of $\mathbf{C}$ that are preserved under normalisation: 
\begin{itemize} 
\item a proof $d$ of $A\seq B$  is interpreted as a morphism $\llbracket d \rrbracket $ from the object $\llbracket A \rrbracket$  to the object $\llbracket B \rrbracket$. There might as well be morphisms from  $\llbracket A \rrbracket$  to  $\llbracket B \rrbracket$ that are not the interpretation of a proof of $A\seq B$ unless the interpretation is proved to be ``fully abstract" (some also say  ``total"). 
\item whenever a poof $d$ of $A\seq B$ reduces to a proof $d'$ (hence a proof of $A\seq B$ as well) by  (the transitive closure) of $\beta$-reduction  or cut-elimination one has $\llbracket d \rrbracket= \llbracket d' \rrbracket $ (hence the name \emph{denotational} semantics). 
\item when there is a terminal object $\un$, a semantic entity in the object (formula)  $X$ of $\mathbf{C}$ can be viewed as a morphism from $\un$ to $X$. Semantic entities include the interpretations of the proofs of $\seq X$ but there might be other semantic entities in $X$, that are not the interpretation of any proof of $\seq X$ when the interpretation is not fully abstract. 
\end{itemize} 
\end{itemize} 
\end{dfn} 
%
%
%
%

Categorical interpretations of  intuitionistic logic
take places in Cartesian closed categories while categorical interpretations of  linear logic  take place in a monoidal closed category --- with  monads for interpreting the modalities (also called exponentials) of full  linear logic.

\subsection{Coherence Spaces}
\label{cohDef}

The category of coherence spaces is a concrete category: objects are (countable) sets endowed with a binary relation, and morphisms are linear maps.  It interprets the proofs up to cut-elimination (or $\beta$ reduction). Coherence spaces are tightly related to linear logic: indeed, linear logic arose from this particular semantics, initially invented to model second order lambda calculus i.e. quantified propositional intuitionistic logic \cite{Gir86}. Coherence spaces are themselves inspired from the categorical work on ordinals by Jean-Yves Girard; they are the binary qualitative domains. 

\begin{dfn}[coherence space] \label{cohSpace} 
A coherence space $A$ is a set $|A|$ (possibly infinite) called the \emph{web} of $A$  whose elements are called \emph{tokens}, endowed with a binary reflexive and symmetric relation called \emph{coherence} on $|A|\times |A|$ written  $\alpha\coh \alpha'[A]$ or  simply $\alpha\coh \alpha'$ when $A$ is clear. 
\end{dfn}   

The following notations are common and useful: 
\begin{verse} 
$\alpha\scoh \alpha'[A]$  iff  $\alpha\coh \alpha'[A]$ and $\alpha\neq \alpha'$

$\alpha\incoh \alpha'[A]$  iff  $\alpha\not\coh \alpha'[A]$ or $\alpha=\alpha'$

$\alpha\sincoh \alpha'[A]$  iff  $\alpha\not\coh \alpha'[A]$ and $\alpha\neq \alpha'$
\end{verse}

A proof of $A$ is to be interpreted by a \emph{clique}  of  the corresponding coherence spaces $A$, a \emph{clique} being a set of pairwise coherent tokens in $|A|$ --- we write $x\in A$ for $x\subset |A|$ and for all $\alpha,\alpha'$ in $x$ one has $\alpha\coh \alpha'$. Observe that for all $x\in A$, if $x'\subset x$ then  $x'\in A$. 

\begin{dfn}[linear morphism]  
A linear  morphism $F$ from $A$ to $B$ is a morphism mapping cliques of $A$ to cliques of $B$ such that: 
\begin{itemize} 
\item For all  $x\in A$ if  $(x'\subset x)$ then  $F(x')\subset F(x)$
\item For every family $(x_i)_{i\in I}$ of pairwise compatible cliques --- that is to say    $(x_i\cup x_j)\in A$ holds for all $i,j\in I$ --- 
 $F(\cup_{i\in I} x_i)= \cup_{i\in I} F(x_i)$.\footnote{The morphism is said to be \emph{stable} when this second condition is replaced with $F(\cup_{i\in I} x_i)= \cup_{i\in I} F(x_i)$ holds more generally for the union of a directed family of cliques of $A$, i.e. $\forall i,j\exists k\ (x_i\cup x_j)=x_k$.}  
\item For all $x,x'\in A$ if $(x\cup x')\in A$ then $F(x\cap x')=F(x)\cap F(x')$ -- this last condition is called \emph{stability}.  
\end{itemize} 
Linear morphisms are maps in the set theoretic sense and they compose as maps. 
\end{dfn}   

Due to the removal of structural rules,  linear logic has two kinds of conjunction: 

$$\begin{prooftree} 
\seq \Gamma, A\quad \seq \Delta, B 
\justifies \seq \Gamma,\Delta, A \ts B
\using \ts
\end{prooftree}
\qquad 
\begin{prooftree} 
\seq \Gamma, A\quad \seq \Gamma, B 
\justifies \seq \Gamma, A \& B
\using \& 
\end{prooftree}
$$ 

Those two rules are equivalent when  contraction and weakening are allowed. 
Multiplicative conjunction splits contexts  as in the $\ts$ rule above while the additive conjunction duplicates context as in the $\&$ rule above.  Regarding  denotational semantics, the web of the coherence space associated with a formula $A * B$ with $*$ a binary \emph{multiplicative} connective  is  the Cartesian \emph{product} of the webs of  $A$ and $B$ i.e. $|A*B|=|A|\times |B|$  --- the web of a binary \emph{additive} connective $A \square B$ is the  \emph{disjoint union} (also called \emph{sum})  of the webs of  $A$ and $B$ i.e.  $|A\square B|=|A|\uplus |B|$. 

\begin{dfn} 
Negation is a unary connective which is both multiplicative and additive:\footnote{Indeed, its web is both the Cartesian product of the web of $A$ and the disjoint union of the web of $A$.} 

\centerline{$|A^\perp|=|A|$ and $\alpha\coh '\alpha[A^\perp]$ iff $\alpha\incoh \alpha'[A]$} 
\end{dfn} 

One may wonder how many binary multiplicative connectives there are,
i.e. how many different coherence relations one may define on $|A|\times |B|$ from the coherence relations on $A$ and on $B$. 

We can limit ourselves to the connectives $*$  that are covariant functors in both $A$ and $B$, i.e. the ones such that 
\begin{itemize} 
\item $\alpha \coh \alpha' [A]$ and $\beta \coh \beta'[B]$ entails $(\alpha,\beta)\coh(\alpha',\beta')[A*B]$ 
\item $\alpha \incoh \alpha' [A]$ and $\beta \incoh \beta'[B]$ entails $(\alpha,\beta)\incoh(\alpha',\beta')[A*B]$ 
\end{itemize} 
Indeed,  there is a negation, hence a contravariant connective in $A$ is a covariant connective in $A^\perp$. Hence when both components are in the $\incoh$ relationship so are the two couples, 
and when they are both coherent, so are the two couples. If a connective $A\star B$ is not like this then $A^\perp\star B$ or $A\star B^\perp$ or $A^\perp\star B^\perp$ is as we wish. 

As it is easily observed, given two tokens $\alpha,\alpha'$ in a coherence space $C$
exactly one of the three following properties hold: 

$$\alpha \sincoh \alpha' \qquad \alpha = \alpha' \qquad \alpha \scoh \alpha'$$  

In order to define a multiplicative connective, 
one should simply say when $(\alpha,\beta)\coh(\alpha',\beta')[A*B]$ holds depending on whether  $\alpha\coh \alpha'[A]$ and $\beta\coh \beta'[B]$ hold. 
Thus defining a binary multiplicative connective is to fill a nine cell table as the ones below, the first column indicates the relation between $\alpha$ and $\alpha'$ in $A$, while the first row indicates  the relation between $\beta$ and $\beta'$ in $B$. 

However if $*$ is assumed to be covariant in both its arguments, seven out of the nine cells are filled, so the only free values  are the ones in the right upper cell, NE=North-East,  and the one in the left bottom cell SW=South-West.  They cannot be $=$ so it makes four possibilities. 

\renewcommand\ne{{\scriptstyle\mathrm{NE}}}
\newcommand\sw{{\scriptstyle\mathrm{SW}}}

$$ 
\tablecoh{A}{B}{*}{\ne?}{\sw?}
$$

If one wants $*$ to be commutative, there are only two possibilities, namely $NE=SW=\scoh$ ($\pa$)  and $NE=SW=\sincoh$ ($\ts$).

$$\tablecoh{A}{B}{\pa}{\scoh}{\scoh} 
\hfill 
\mbox{\qquad and \qquad} 
\hfill 
\tablecoh{A}{B}{\ts}{\sincoh}{\sincoh} 
$$

However if we do not ask for the connective $*$ to be commutative we have a third \emph{non-commutative} connective $A\bef B$ and a fourth connective $A\triangleright B$ which is simply $B\bef A$.  

$$
\tablecoh{A}{B}{\bef}{\sincoh}{\scoh} 
\mbox{\qquad and \qquad} 
\tablecoh{A}{B}{\triangleright}{\scoh}{\sincoh} 
$$ 

\begin{dfn}[multiplicative connectives] 
The coherence relation for the multiplicative  connectives are defined as follows: 

$(\alpha,\beta)\coh(\alpha',\beta')[A\ts B]$ iff $\alpha\coh \alpha'[A]$ and $\beta\coh \beta'[B]$

$(\alpha,\beta)\coh(\alpha',\beta')[A\pa B]$ iff $\alpha\coh \alpha'[A]$ or  $\beta\coh \beta'[B]$

$(\alpha,\beta)\coh(\alpha',\beta')[A\bef B]$ iff $\alpha\scoh \alpha'[A]$ or  ($\alpha=\alpha'$ and $\beta\coh \beta'[B]$) 
\end{dfn}

\begin{dfn}[trace of a linear morphism]
The linear morphisms from $A$ to $B$ written $\Hom(A,B)$ can be represented by the coherence space  $A\multimap B=A^\perp\pa B$ \footnote{This internalisation of $\Hom(A,B)$ makes the category monoidal closed, but not Cartesian closed because the associated conjunction, namely $\ts$ is not a product.}

$$ 
\begin{array}{l||c|c|c|}
A \multimap B & \sincoh & = & \scoh \\ \hline \hline 
\sincoh & \scoh & \scoh & \scoh  \\ \hline 
= & \sincoh & = & \scoh \\ \hline 
\scoh & \sincoh  & \sincoh & \scoh \\ \hline 
\end{array} 
$$ 

Linear morphisms are in a one-to-one correspondence with cliques of $A\multimap B$.  Given a clique $f\in (A\multimap B)$ the map $F_f$ from cliques of $A$ to cliques of $B$ defined by $F_f(x)=\{\beta\in |B|\ |\ \exists \alpha\in x\ \mathrm{such\ that}\ (\alpha,\beta)\in f\}$ is a linear morphism. Conversely, given a linear morphism $F$, the set $\{(\alpha,\beta)\in |A|\times |B|\ |\ \beta\in F(\{\alpha\})\}$ called the \emph{trace} of $F$  is a clique of $A\multimap B$. 
\end{dfn}

\begin{prop}
The connective $\bef$ is (1) non commutative, (2) associative, (3) self-dual and (4) lays in between $\otimes$ and $\pa$. 

\end{prop} 

\begin{proof} 
\begin{enumerate} 
\item $\bef$ is non commutative. 
From the definition of the coherence space $A\bef B$ it is clear that there is no canonical isomorphism between $A\bef B$ and $B\bef A$, hence   $A\bef B$ and $B\bef A$ are not isomorphic in general. 
\item $\bef$ is associative. 
The set $\{((\alpha,(\beta,\gamma)),((\alpha,\beta),\gamma))\ |\ \alpha\in |A|,\beta\in |B|, \gamma\in |C|\}$  defines a linear isomorphism from $A\bef (B\bef C)$ to 
$(A\bef B)\bef C$.  
\item $\bef$ is self-dual, $(A\bef B)^\perp\equiv (A^\perp)\bef (B^\perp)$. 
Given two different tokens $(\alpha,\beta)$ and $(\alpha',\beta')$ in $|A|\times|B|$, observe that: 
\begin{enumerate} 
\item 
$(\alpha,\beta)\scoh(\alpha',\beta')[A\bef B]$ means 
($\alpha=\alpha'$ and $\beta\scoh\beta'[B]$)   or $\alpha\scoh\alpha'[A]$
\item 
$(\alpha,\beta)\scoh(\alpha',\beta')[A^\perp\bef B^\perp]$ means 
($\alpha=\alpha'$ and $\beta\sincoh\beta'[B]$)   or $\alpha\sincoh\alpha'[A]$
\end{enumerate} 

Given that those two tokens $(\alpha,\beta)$ and $(\alpha',\beta')$ are different, either: 
\begin{itemize} 
\item 
If $\alpha\neq\alpha'$ then either 
\begin{verse} 
$\alpha\scoh\alpha'[A]$,  1 holds and 2 does not hold or 

$\alpha\sincoh\alpha'[A]$,  2 holds and 1 does not hold. 
\end{verse} 
\item If $\alpha=\alpha'$ then $\beta\neq\beta'$ and 
either 
\begin{verse} 
$\beta\scoh\beta'[B]$, 1 holds and 2 does not hold 
or 

$\beta\sincoh\beta'[B]$ 2 holds and 1 does not hold. 
\end{verse} 
\end{itemize} 
Consequently, any two different tokens $(\alpha,\beta)$ and $(\alpha',\beta')$ in $|A|\times |B|$ are either strictly coherent in $(A\bef B)$ or are strictly coherent in $A^\perp\bef B^\perp$. Therefore 
$(A\bef B)^\perp\equiv (A^\perp)\bef (B^\perp)$. 

\item $(A \otimes B) \multimap  (A\bef B) \multimap  (A\pa B)$
The set $\{((\alpha,\beta),(\alpha,\beta))\ |\ \alpha\in |A|, \beta\in |B|\}$
defines a linear morphism from $A\ts B$ to $A\bef B$ and the very same set of pairs of tokens also defines a linear morphism from $A\bef B$ to $A\pa B$. 
\end{enumerate} 
\end{proof} 

The definition of the coherence spaces associated with $A\bef B$ and $A\pa B$ can be generalised to series parallel  (\sp) partial orders of formulas, a well-known class of  finite partial orders defined as follows. Given two partial orders $O_1\subset |O_1|^2$ and $O_2\subset |O_2|^2$ on two disjoint sets  $|O_1|$ and $|O_2|$, one can defined two partial orders on 
$|O_1|\uplus|O_2|$: 

\begin{itemize} 
\item 
$O_1 \pa O_2$  the parallel composition of  $O_1$ and $O_2$ defined by $O_1\uplus O_2 \subset (|O_1|\uplus|O_2|)^2$ 
is the disjoint union of the two partial orders $O_1$ and $O_2$ viewed as two sets of couples 
\item 
 $O_1\bef O_2$ the series composition   of  $O_1$ and $O_2$  
 defined by $O_1\uplus O_2  \uplus (|O_1|\times|O_2|) \subset (|O_1|\uplus|O_2|)^2$
\end{itemize} 

The class of series-parallel partial orders is the the smallest class of binary relations defined from the one element orders and closed by two binary operations $\pa$ and $\bef$  on partial orders defined above. 

\begin{dfn}[sp-order of coherence spaces]\label{spcoh}  
Given  an \sp-order $O$ on $\{1,\ldots,n\}$ and a multiset of formulas $\{A_1,\ldots,A_n\}$   (one may have $A_i=A_j$ as formulas for $i\neq j$) one may consider  the formula $O(A_1,\ldots,A_n)$ defined from $A_1,\ldots,A_n$ using only  the  $\pa$ and $\bef$ connectives, following the order $O$: series composition in $O$ corresponds to the $\bef$ connective  in $O(A_1,\ldots,A_n)$
and parallel composition in $O$ corresponds to the $\pa$ connective  in $O(A_1,\ldots,A_n)$. 
One can define directly the coherence space corresponding to the formula $O(A_1,\ldots,A_n)$ as follows: 
\begin{quote} 
\noindent web: $|A_1|\times \cdots \times |A_n|$ 

\noindent 
strict coherence: $(\alpha_1,\ldots, \alpha_n)\scoh (\alpha'_1,\ldots, \alpha'_n)[O[A_1,\ldots,A_n]]$\\ 
whenever there exists $i$ such that  $\alpha_i\scoh \alpha'_i$
and for every  $j<_O i$  one has 
$\alpha_j = \alpha'_j$. 
\end{quote} 
\end{dfn} 

Coming back to what a categorical interpretation is, let us define precisely the category where our categorical interpretation takes place: 

\begin{dfn}[The category $\mathbf{Coh}$]\label{Coh} 
The category of coherence spaces is defined by its objects which are coherence spaces and its linear morphisms which are linear maps. It contains 
constructions corresponding to negation and to the connectives $\pa,\ts,\bef$. 
Linear maps from $A$ to $B$ can be internalised as their trace, that are object of the coherence space $A \multimap B$. 
\end{dfn}

Linear logic is issued from coherence semantics,  
and consequently coherence semantics is close to linear logic syntax. Coherence spaces may even be turned into a fully abstract (a.k.a. total) model in the multiplicative case (without before), see \cite{Loader94}. 
The before connective and pomset logic are also issued from coherence semantics, so let us define precisely the categorical interpretation of pomset logic, the logic that the flag modality is extending.

\subsection{A Sound and Faithful Interpretation of Pomset Proof-Nets in Coherence Spaces}
\label{cohIntPN}

We cannot present here pomset proof-nets in full detail, but they are thoroughly presented in the recent paper \cite{Retore2021Lambek}. Nevertheless we can briefly present handsome proof-nets for pomset logic. This view of proof-nets was initially introduced in \cite{Ret96entcs} for MLL,  and in \cite{Ret98roma} for pomset logic. 

De Morgan laws hold for pomset logic:

$$
\begin{array}{rcl} 
(A^\perp)^\perp&=&A\\ 
(A\pa B)^\perp&=&(A^\perp\ts B^\perp)\\ 
(A\bef B)^\perp&=&(A^\perp\bef B^\perp)\\ 
(A\ts B)^\perp&=&(A^\perp\pa B^\perp)\\ 
\end{array}
$$ 

Because of the De Morgan laws above, as in classical logic,  
every formula does have an equivalent negation normal form (in which negation only appears on atoms), 
and proof-net syntax only use negation normal forms. 
The set of pomset formulas in negation normal form $F$ over a set of propositional variables $P$ is defined as follows: 

$$F {}{:}{=} P \ |\ P^\perp \ |\ F \pa F\ |\ F \bef F\ |\ F \ts F$$

The  \emph{directed cograph} or \emph{dicograph} for short $\dicog(F)$ associated with a formula $F$ is defined as follows. Given a formula $F$ we say that two occurrences of atoms $a$ and $b$ (elements of $P\uplus P^\perp$) in  $F$ meet on the connective $*$ ($*$ being one of $\pa,\bef,\ts$) whenever $F$ is $H[G[a]*G[b]]$ where $E[]$ denotes a formula with a  hole --- beware that the order matters, $a$ and $b$ meet on $*$ does not mean that $b$ and $a$ meet on $*$, one connective being non commutative. The arcs of $\dicog(F)$ are defined as follows: 
\begin{itemize} 
\item 
there is an $R$ (for Red or Regular, when printed in black and white) arc from $a$ to $b$ and no $R$ arc from $b$ to $a$ whenever  $a$ and $b$ meet on a ``before", 
\item 
there is an $R$ edge (that is a pair of opposite $R$ arcs) between $a$ and $b$ whenever   $a$ and $b$ meet on a tensor,
\item there is  no $R$ arc from $a$ to $b$ nor from $b$ to $a$ when $a$ and $b$ meet on a par
\end{itemize}

Several distinct formulas may describe the same dicograph: in this case the formulas only differ up to the associativity and commutativity of $\pa,\ts$ and to the associativity of $\bef$, e.g. $(A\otimes (B\bef (C\bef D)))\otimes E$ and 
$(E\otimes A) \otimes ((B\bef C)\bef D)$ define the same dicograph. Dicographs are directed graphs that are the dicograph of some  formula. Cographs are dicographs with no directed edges, i.e. the dicograph corresponding to a formula without any $\bef$ connective. Series-parallel partial orders (already mentioned in subsection \ref{cohDef})  are the dicographs with no undirected edges, corresponding to formulas with no $\otimes$ connectives. For more details on those structures and their characterisation by absence of certain subgraphs see \cite{BGR97} or  the more recent book chapters \cite{Retore2021Lambek,GuoSurmacs2018dicographs}

A \emph{pomset proof structure} with conclusion $F$ is an edge bicoloured graph (Red or Regular arcs and Blue or Bold edges) overs the occurrences of atoms of $F$ consists in the  $R$ dicograph $\dicog(F)$ and a set of $B$ edges satisfying: no two $B$ edges are adjacent, every atom (vertex) of $F$ is incident to a (unique) $B$ edge, and the end vertices of $B$ edges are dual atoms $a$ and $a^\perp$. Notice that the conclusion of a proof structure contains the same number of occurrences of $a$ and of occurrences of $a^\perp$ for each propositional variable $a$. The conclusion  $F$ of a proof structure $\pi$ way be written in several ways as $G[C_1,\ldots,C_n]$ with $G$ only using the connectives $\pa,\bef$: in this case one may also say that the conclusions of $\pi$ are $C_1,\ldots,C_n$ ordered by the series parallel partial order described by $G$.

A \emph{pomset proof-net} or a correct pomset proof structure is a proof structure when it satisfies the \emph{correctness criterion}: every elementary circuit (directed cycle) alternating $R$ and $B$ edges (\ae-circuit for short) contains an $R$ chord (an $R$ edge or arc not on the circuit but whose end vertices are on the circuit).

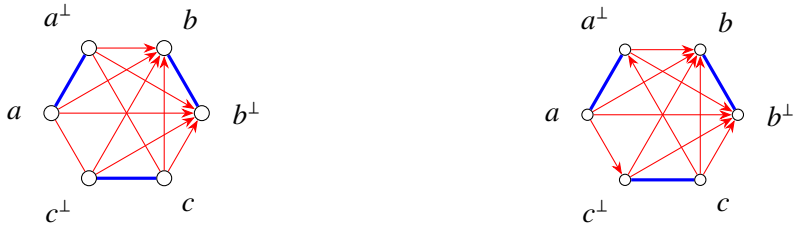
\begin{figure}[t]
\begin{tikzpicture}[PS]
\tikzset{vertex/.append style={inner sep=2pt}} 
  \vertex[orientation=180,on circle]{a};
    \vertex[orientation=120,neg,on circle]{a};
   \vertex[orientation=60,on circle]{b};
   \vertex[orientation=0,neg,on circle]{b};
   \vertex[orientation=-60,on circle]{c};
   \vertex[orientation=-120,neg,on circle]{c};
   \blink{c}; 
   \blink{b};
   \blink{a};
    \draw[Rlink,->] (a) -- (b);
    \draw[Rlink,->] (acomp) -- (b);
    \draw[Rlink,->] (c) -- (b);
    \draw[Rlink,->] (ccomp) -- (b);
   \draw[Rlink,->] (a) -- (bcomp);
    \draw[Rlink,->] (acomp) -- (bcomp);
    \draw[Rlink,->] (c) -- (bcomp);
    \draw[Rlink,->] (ccomp) -- (bcomp);
    \draw[Rlink] (a) -- (ccomp);
     \draw[Rlink] (acomp) -- (c);
 \end{tikzpicture} 
\hspace*{3cm}
 \begin{tikzpicture}[PS]
  \vertex[orientation=180,on circle]{a};
    \vertex[orientation=120,neg,on circle]{a};
   \vertex[orientation=60,on circle]{b};
   \vertex[orientation=0,neg,on circle]{b};
   \vertex[orientation=-60,on circle]{c};
   \vertex[orientation=-120,neg,on circle]{c};
   \blink{c}; 
   \blink{b};
   \blink{a};
    \draw[Rlink,->] (a) -- (b);
    \draw[Rlink,->] (acomp) -- (b);
    \draw[Rlink,->] (c) -- (b);
    \draw[Rlink,->] (ccomp) -- (b);
   \draw[Rlink,->] (a) -- (bcomp);
    \draw[Rlink,->] (acomp) -- (bcomp);
    \draw[Rlink,->] (c) -- (bcomp);
    \draw[Rlink,->] (ccomp) -- (bcomp);
    \draw[Rlink,->] (a) -- (ccomp);
     \draw[Rlink,<-] (acomp) -- (c);
 \end{tikzpicture} 
\caption{Two incorrect handsome proof structures  namely 
$((a\ts c^\perp)\pa (a^\perp \ts c))\bef (b\pa b^\perp)$ and 
$((a\bef c^\perp)\pa (c \bef a^\perp))\bef (b\pa b^\perp)$.
Both are incorrect because of the 
 chordless \ae-circuit: {$a,c^\perp,c,a^\perp,a$}. Red or Regular arcs represent the conclusion formula (dicograph) while Blue or Bold edges represent the axioms. }
 \end{figure}

\begin{figure}[t]
\begin{tikzpicture}[PS]
  \vertex[orientation=180,on circle]{a};
    \vertex[orientation=120,neg,on circle]{a};
   \vertex[orientation=60,on circle]{b};
   \vertex[orientation=0,neg,on circle]{b};
   \vertex[orientation=-60,on circle]{c};
   \vertex[orientation=-120,neg,on circle]{c};
   \blink{c}; 
   \blink{b};
   \blink{a};
    \draw[Rlink,->] (a) -- (b);
    \draw[Rlink,->] (acomp) -- (b);
    \draw[Rlink,->] (c) -- (b);
    \draw[Rlink,->] (ccomp) -- (b);
   \draw[Rlink,->] (a) -- (bcomp);
    \draw[Rlink,->] (acomp) -- (bcomp);
    \draw[Rlink,->] (c) -- (bcomp);
    \draw[Rlink,->] (ccomp) -- (bcomp);
     \draw[Rlink] (a) -- (c);
    \draw[Rlink] (a) -- (ccomp);
     \draw[Rlink] (acomp) -- (c);
    \draw[Rlink] (acomp) -- (ccomp);   
 \end{tikzpicture} 
\hspace*{3cm}
 \begin{tikzpicture}[PS]
  \vertex[orientation=180,on circle]{a};
    \vertex[orientation=120,neg,on circle]{a};
   \vertex[orientation=60,on circle]{b};
   \vertex[orientation=0,neg,on circle]{b};
   \vertex[orientation=-60,on circle]{c};
   \vertex[orientation=-120,neg,on circle]{c};
   \blink{c}; 
   \blink{b};
   \blink{a};
    \draw[Rlink,->] (a) -- (b);
    \draw[Rlink,->] (acomp) -- (b);
    \draw[Rlink,->] (c) -- (b);
    \draw[Rlink,->] (ccomp) -- (b);
   \draw[Rlink,->] (a) -- (bcomp);
    \draw[Rlink,->] (acomp) -- (bcomp);
    \draw[Rlink,->] (c) -- (bcomp);
    \draw[Rlink,->] (ccomp) -- (bcomp);
    \draw[Rlink,->] (a) -- (ccomp);
     \draw[Rlink,->] (acomp) -- (c);
 \end{tikzpicture} 
\caption{Two correct handsome proof structures (or proof-nets) 
$((a\pa a^\perp)\ts (c^\perp \pa c))\bef (b\pa b^\perp)$ and 
$((a\bef c^\perp)\pa (a^\perp \bef c))\bef (b\pa b^\perp)$ --- none of them contains a chordless \ae-circuit. 
Red or Regular arcs represent the conclusion formula (dicograph) while Blue or Bold edges represent the axioms. 
}
\end{figure}
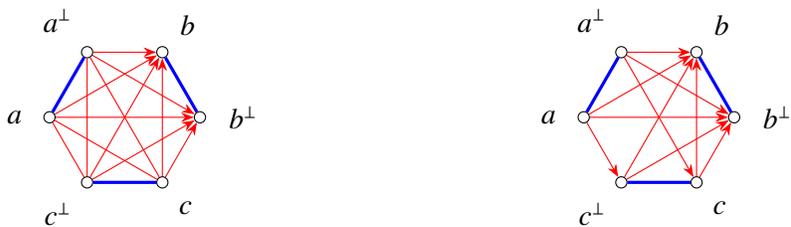

It  is possible to interpret a  proof-net and even a simple proof structure 
with conclusion $T$ (a formula or a dicograph of atoms) as a set of tokens of the corresponding coherence space $T$. 
So the fact that the interpretation should be a clique in the coherence space $\llbracket T\rrbracket$ and also cut-elimination guide the 
design of the deductive systems for pomset logic \cite{Retore2021Lambek}. 

Computing the semantics of a proof-net  is rather easy, using Girard's experiments --- we rather define experiments from axioms to conclusions as in \cite{Ret94rc,Ret97} but it makes no difference. 
We define the interpretation of a proof structure, which is not necessarily a proof-net as a set of tokens of the web of the conclusion formula. Assume the axioms of a proof structure is the set of  $B$-edges $B=\{a_i\edge a_i^\perp|1\leq i \leq n\}$ and that each of the $a_i$ has a corresponding coherence space also denoted by $a_i$. 
For each $a_i$ choose a token $\alpha_i\in |a_i|$. If the conclusion of the proof structure is a dicograph $T$ then replacing the two occurrences  $a_i$ and  $a_i^\perp$ that are connected with an axiom (linked with a $B$ edge) with $\alpha_i$ yields a term, which when converting $x * y$ (with $*$ being one of the connectives, $\pa,\bef,\ts$) with $(x,y)$, yields a 
token in the web of the coherence space associated with $T$ --- this token in $|T|$ is called the \emph{result} of the  experiment.

Given a normal (cut-free) proof structure $\pi$ with conclusion $T$ 
the interpretation $\llbracket \pi \rrbracket$ of the normal proof structure $\pi$
is the set of all the results of the experiments on $\pi$. As initially shown for proof-nets with links in \cite{Ret94rc} or for handsome proof nets in \cite{Retore2021Lambek}: 

\begin{thm}\label{semcorrect} 
A proof structure $\pi$ with conclusion $T$ is a proof-net (contains no chordless \ae-circuit) if 
and only if its interpretation $\llbracket \pi \rrbracket$ is a clique of the coherence space  $T$ 
--- rememember that only cliques of $T$ can be the interpretation of some proof of $T$. 
\end{thm}

When $\pi$ is not normal, i.e. includes cuts that are formulas $K\ts K^\perp$, not all  experiments succeed and provide results:  an experiment is said to succeed when in every cut $K \otimes K^\perp$ (a cut may be viewed as a conclusion $K\otimes K^\perp$) the value $\alpha$ on an atom $a$ of the cut $K$ is the same as the value on the corresponding atom $a^\perp$ in  $K^\perp$. Otherwise the experiment fails and has no result.  The set of the results of all succeeding experiments of a proof \emph{net}  $\pi$ is a clique of the coherence space $T$. It is the interpretation $\llbracket \pi \rrbracket$ of the normal proof-net $\pi$. We can prove, that whenever $\pi$ reduces to $\pi'$ by cut elimination $\llbracket \pi\rrbracket=\llbracket \pi' \rrbracket$, see e.g. \cite{Ret94rc}.

\section{The ``Flag" Modality in Coherence Spaces}

The construction of the coherence space $\flag A$ is somehow inspired by the fact that the web $|!A|$ of $!A$ consists in the finite sets of elements in the web $|A|$. 
The intuition is that  the web of $\flag A$ consists in linearly ordered copies of the web of $A$,  the underlying order being  equivalent to two ordered copies of it, so that $\flag A$ is isomorphic to $\flag A \bef \flag A$, i.e. the web  $|\flag A|$ is $\cdots \bef |A|\bef |A|\bef |A|\bef |A|\bef |A|\bef \cdots$. For the order to be equivalent to two copies of it, we first succeeded with $\mathbb{Q}$ copies of $A$, but elements of $|\flag A|$ were lacking a finite description. Then, following a suggestion by Achim Jung, we looked at $2^\omega$ copies of $|A|$ and even better at the continuous functions from $2^\omega$ to $|A|$ endowed with the discrete topology on $|A|$ and that way  elements of the web $|\flag A|$ of $\flag A$ have a finite description.

%
%

\subsection{Remarks on the Continuous Functions from the Cantor Space to a Discrete Space}

Let us write  $\deux$ for $\{0,1\}$ and  $\deux^{*}$ for the set of finite words on $\deux$, including the
empty word, $\deux^{\omega}$ for the set of infinite words on $\deux$. Letters like $w, v, u$ range over $\deux^{\omega},$ while $m$  range over $\deux^{*}$.  The expression $w=m(m')^{\omega}$ stands for $w=m m' m' m' m' m' \cdots$

The set $\deux^{\omega}$ of infinite words on $\deux$ is assumed to be endowed with: 
\begin{itemize} 
\item 
 the usual total lexicographical order on infinite strings defined by: 
$$w_1<w_2 \ \mbox{iff}\ \exists m\in\deux^*\ \exists w'_1, w'_2\in \deux^\omega\  w_1= m0w'_1\ and\ w_2=m1w'_2$$ 
\item 
the usual product topology generated by the basis of clopen sets $(U_m)_{m\in \deux^*}$ with $$U_m=\{w\in \deux^\omega | \exists w'\in \deux^\omega\ w=mw'\}$$ 
\end{itemize}

A well-known result on $\deux^\omega$ is: 
\begin{prop}\label{minimum} 
A clopen set $U$ of $\deux^\omega$  always is a finite union of base clopen sets, and thus a clopen set always has a minimum element (and a maximal element as well). 
\end{prop} 

\begin{proof}
As $U$ is open, 
it is a union of base clopen set say $(O_i)_{i\in I}$. But as $U$ is closed (clopen) in a compact, $U$ is compact, and from the covering of $U$ by the $(O_i)_{i\in I}$, there exists  a finite covering $(O_{i_k})_{k\in K}$  of $U$ by base clopen sets with $K$ finite. Base clopen sets do have a minimum element, and the minimum of those finitely many minimums is the minimum of $U$. 
\end{proof}

Given a continuous function $f$  from $\deux^\omega$  to $M$ endowed with the discrete topology,  
$f^{-1}\{a\}$ is a clopen set of $\deux^\omega$  and $\cup_{a\in M} f^{-1}\{a\}= \deux^\omega$.

Because $\deux^\omega$ is compact a finite open cover can be extracted from $\cup_{a\in M} f^{-1}\{a\}= \deux^\omega$ 
so there exists a finite number $n$ such that 
$\cup_{1\leq i \leq n} f^{-1}\{a_i\}= \deux^\omega$, and replacing each clopen $f^{-1}\{a_i\}$ by a finite union of base clopen sets, we obtain a finite cover by base clopen sets $U_k$, $1\leq k \leq p$, with $f$ constant on each $U_m$. 

From this one easily obtains: 

\begin{prop}  \label{repfctcont} 
$(\gt_{M} \text { generic trees on } M)$ The set $\gt_{M}$ of \emph{continuous functions from
$\deux^{\omega}$ to a set $M$} (discrete topology) is in a one-to-one correspondence with the set of
finite binary trees on $M$ such that any two sister leaves have distinct labels.

Thus, an element $f$ of $\gt_{M}$ may either be described:
\begin{enumerate} 
\item \label{couple} 
As a finite set $\left\{(m_{1}, a_{1}), \ldots,(m_{k}, a_{k})\right\} \subset \mathcal{P}_{\text {fin}}(\deux^{*} \times M)$ satisfying:
\begin{enumerate} 
\item  \label{existecouple} $\forall w \in \deux^{\omega}\ \exists ! i\leq k\   \exists w' \in \deux^{\omega} \quad  w=m_{i} w'$
\item \label{unicouple} $\forall i, j\leq k\  \left[\exists m \in \deux^{*}\  m_{i}=m 0 \text { and } m_{j}=m 1\right] \Rightarrow a_{i} \neq a_{j}$ 
\end{enumerate}
In this formalism $f(w)$ is computed as follows: applying $(a),$ there exists a unique $i$ such that  $w=m_{i} w'$ for some $w^\prime$, let $f(w)=f(m_{i}w')=f(m_i(0)^\omega)=f(m_i(1)^\omega)=a_i$
\item \label{termesnormaux} As the normal form of a term of the following grammar where $\underline{M}$ stands for $\{\underline{x} | x\in M\}$\footnote{In order to distinguish between the constant function $\underline{a}$ mapping every word of $\deux^\omega$ onto $a\in M$ and $a\in M$.}  :
$$
\mathcal{T}_{M}:: \underline{M} |\left\langle\mathcal{T}_{M}\, \mathcal{T}_{M}\right\rangle
$$
where the reduction is $\forall x \in M$  
$t\left[\langle\underline{x}\, \underline{x}\rangle\right] \longrightarrow t[\underline{x}]$ where $t[u]$ with  $u$ in $\mathcal{T}_{M}$ means
a term of $\mathcal{T}_{M}$ having an occurrence of the subterm $u \in \mathcal{T}_{M}$.
In this formalism $f(w)$ is computed as follows:
$$ 
\begin{array}{rclcrcl}
f&=&\underline{a} &:& f(w)&=&a\\[1ex] 
f&=&\left\langle t_{0}\, t_{1}\right\rangle &:& f(0 w)&=&t_{0}(w)\\ 
&&&& f(1 w)&=&t_{1}(w)\\
\end{array}
$$
\end{enumerate} 
\end{prop} 

\paragraph{Example:} Let $M=\{a, b, c\}$ Here are the three description of the same element of $\gt_{M}$ :
\begin{enumerate} 
\setcounter{enumi}{-1}
\item  As a function
$$
\begin{array}{rclcrcl} 
f(000 w)&=&a&\hspace*{2em}&f(100 w) &=&a\\
f(001 w)&=&a&&f(101 w) &=&b\\ 
f(010 w)&=&a&&f(110 w) &=&a\\ 
f(011 w)&=&a&&f(111 w) &=&b\\ 
\end{array}
$$
\item  as a finite set of pairs $\left\{(m_{i}, a_{i})\ |\  m_{i} \in \deux^{*} \text { and } a_{i} \in M\right\}$ :
$$
f=\{(0, a),(100, a),(101, b),(110, a),(111, b)\}
$$
\item  as a normal term of $\mathcal{T}_{M}: f=\langle\underline{a}\langle\langle\underline{a}\, \underline{b}\rangle\langle\underline{a}\, \underline{b}\rangle\rangle\rangle$
such a term is the normal form of, e.g. $\langle\langle\underline{a}\, \underline{a}\rangle\langle\langle\underline{a} \underline{b}\rangle\langle\langle\underline{a} \underline{a}\rangle \langle \underline{b}\, \underline{b}\rangle \rangle\rangle\rangle$
\end{enumerate}

Here is one more easy remark: 

\begin{prop}\label{diffegalsup} 
Let $f, g$ be two functions in $\gt_{M}$. If $f \neq g$ then there exists $w \in \deux^{\omega}$ such that
$$
f(w) \neq g(w) \text { and } \forall w' {<}w\quad f(w')=g(w')
$$
\end{prop}

\begin{proof}  The product of the discrete topological space $M$ by itself is the discrete topological space over $M \times M$.  Hence the function $\Delta$ from $M \times M$ to $\deux$ defined by $\Delta(x, y)=1$ iff $x=y$ is continuous. The function $(f, g)$ from $\deux^{\omega}$ to $M \times M$ defined by $(f, g)(w)=(f(w), g(w))$
is continuous, because we use the product topology on $M \times M$.  Therefore $(\Delta \circ(f, g))^{-1}(0)$ is a clopen set, which has a lowest element $w$ ending  with an infinite sequence of $0$ (as observed in proposition \ref{minimum}), i.e. $w=m(0)^{\omega}$.  Thus this $w$ enjoys $f(w) \neq g(w)$ and $f(v)=g(v)$ whenever $v{<} w$. 
\end{proof} 

\subsection{The ``Flag" Modality} 
As said in previous section, after a first solution with rational numbers in $[0,1[$, following a suggestion by Achim Jung, we arrived to the following solution: 

\begin{dfn}  
Let $A$  be a coherence space. We define $\flag  A$  as follows:
\begin{description} 
\item[web] $|\flag A|=\gt_{|A|}$ the set of continuous functions from $\deux^{\omega}$ to $|A|$ (discrete topology). 
\item[coherence] Two functions $f$ and $g$ of $\gt_{|\mathrm{A}|}=|\flag  A|$ are said to be strictly coherent whenever
$$
\exists w \in \deux^{\omega} f(w) \scoh g(w)[A] \text { and } \forall w' {<}  w\enspace f(w')=g(w')
$$
\end{description} 
\end{dfn}  

Observe that the above definition generalises to the infinite case $A_1\bef A_2 \bef \cdots \bef A_n$ with $A_i=A$ for all $i$ 
as defined in definition \ref{spcoh} of the coherence space  associated with an \sp ordered collection of coherence space. 

\subsection{Properties} 
The following is clear from proposition \ref{repfctcont}: 

\begin{prop} 
 (denumerable web) If $|A|$ is denumerable, so is $|\flag  A|$.
 \end{prop} 
 
And next come a key property:

\begin{prop}  (self-duality) The modality $\flag$ is self-dual, i.e. $(\flag  A)^{\perp} \equiv \flag (A^{\perp})$
\end{prop} 

\begin{proof} 
Those two coherence spaces obviously have the same web. Hence it is equivalent to show that, given two distinct tokens $f, g$ in $|\flag  A|$ either $f\scoh g[\flag  A]$ (*) holds or  $f \scoh g \left[\flag (A^{\perp})\right]$ (**) holds. 
If $f \neq g$, then,  because of previous proposition \ref{diffegalsup},   there exists an infinite word $w$ in $\deux^{\omega}$ such that : 

$$f(w) \neq g(w) \mathrm{\ and\ }\forall w' {<}  w\enspace  f(w')=g(w').$$ 

Therefore, 
\begin{itemize} 
\item 
either [$f(w) \scoh g(w)[A]$ and $\forall w'{<} w\enspace  f(w')=g(w')$] holds 
\item 
or [$f(w) \scoh g(w)\left[A^{\perp}\right]$ and $\forall w'{<} w\enspace f(w')=g(w')$] holds.  
\end{itemize} 
This is component-wise equivalent to the expected exclusive disjunction (*) or (**).  
\end{proof}

\begin{prop}  (contraction isomorphism) There is a canonical linear isomorphism 
$\flag  A \leftspoon\!\rightspoon \flag A\bef \flag  A$
\end{prop}

\begin{proof} 
Consider the following subset of $|\flag  A| \times|\flag  A\bef \flag  A|$:
$$
\mathcal{C}=\left\{(h,(h_{0}, h_{1})) | \forall w \in \deux^{\omega}\  h(0 w)=h_{0}(w) \text { and } h(1 w)=h_{1}(w)\right\}
$$
Let us see that it is the trace of a linear isomorphism between $\flag  A$ and $\flag  A\bef \flag  A$.

\noindent Firstly, $\mathcal{C}$ clearly defines a bijection, between the webs $|\flag A|$ and $|\ \flag  A\bef \flag  A\ |=|\flag  A| \ \times\ |\flag  A|$.

\noindent Secondly, let us see that, given $(h,(h_{0}, h_{1}))$ and $(g,(g_{0}, g_{1}))$, both in $\mathcal{C}$ we have
$$
(1):h\scoh g[\flag  A]\  \Longleftrightarrow\ (h_{0}, h_{1})\scoh (g_0, g_1)[\flag  A\bef \flag  A]:(2)
$$

\noindent $(1) \Longrightarrow (2)$ 
We assume that $h \scoh g[\flag  A]$, i.e. that $\exists w \in \deux^{\omega}\ h(w)\scoh g(w)$ and $\forall v{<} w\quad h(v)=g(v)$. 

\noindent 
In each of the two possible cases, $w=0 w'$ or $w=1 w'$, let us we show that  $(h_{0}, h_{1}) \scoh (g_{0}, g_{1})[\flag  A\bef \flag  A]$

\noindent\textit{Unless otherwise specified coherence relations are assumed to take place in the coherence space $A$.} 
\begin{enumerate} 
\setcounter{enumi}{-1} 
\item 
If $w=0 w'$ we have $h_{0} \scoh  g_{0}[\flag  A]$: 
\begin{itemize} 
\item 
$h_0(w')\scoh g_0(w')$ since $h_0(w')= h(0w')=h(w)$,\hfill  $g(w)=g(0w')=g_0(w')$ and  $h(w)\scoh g(w)$. 
\item 
 $h_0(v')=g_0(v')$ for all $v'{<}  w'$; indeed, $0v' < 0w'=w$ hence $h_0(v')=h(0v')=g(0v')=g_0(v')$. 
\end{itemize} 
\item 
If $w=1w'$ then $h_1\scoh g_1[\flag A]$ and $h_0=g_0$: 
\begin{itemize} 
\item 
 $h_1\scoh g_1[\flag A]$
 \begin{itemize}
 \item $h_1(w')\scoh g_1(w')$  since  $h_1(w')=h(1w')=h(w)$, \hfill $g_1(w')=g(1w')=g(w)$, \hfill $h(w)\scoh g(w)$. 
 \item $h_1(v')=g_1(v')$ for all $v'{<} w'$; indeed, $h_1(v')=h(1v')=g(1v')=g_1(v')$ since $1v'{<}1w'=w$. 
 \end{itemize} 
 \item $h_0=g_0$ since $h_0(u)=h(0u)=g(0u)=g_0(u)$ because $0u<1w'=w$. 
 \end{itemize} 
 \end{enumerate}

\noindent $(2) \Longrightarrow (1)$ 
We assume that $(h_{0}, h_{1}) \scoh (g_{0}, g_{1})[\flag  A\bef \flag  A]$ i.e. that either 
$h_{0} \scoh g_{0}[\flag A]$ or  ($h_{0}=g_{0}$ and $h_{1} \scoh  g_{1}[\flag A]$).  We show that in both cases we have $h \scoh  g [\flag  A]$
\begin{enumerate} 
\setcounter{enumi}{-1}
\item If $h_{0}\scoh g_{0} [\flag A]$ then there exists $w'$ such that $h_0(w')\scoh g_0(w')$ and $h_0(v')=g_0(v')$ for all $v'{<} w'$ and $h\scoh g [\flag A]$. 
Indeed: 
\begin{itemize} 
\item $h(0w')\scoh g(0w')$ because $h(0w')=h_0(w')$,\enspace $g(0w')=g_0(w')$ and $h_0(w')\scoh g_0(w')$.   
\item for all $u< 0w'$, one has $h(u)=g(u)$; indeed, if $u<0w'$ then $u=0u'$ with $u'<w'$, so $h(u)=h(0u')=h_0(u')$, $g(u)=g(0u')=g_0(u')$ and $h_0(u')=g_0(u')$ because $u'<w'$. 
\end{itemize} 
\item If $h_{1} \scoh g_{1}[\flag A]$ and $h_{0}=g_{0}$ then there exists $w'$ such that $h_{1}(w') \scoh  g_{1}(w')$ and $h_{1}(v')=g_{1}(v')$ for all $v'<w'$,
and for all $u$, $h_0(u)=g_0(u)$. We have $h \scoh  g [\flag  A]$: 
\begin{itemize} 
\item $h(1w')\scoh g(1w') [A]$; indeed $h(1w')=h_1(w')$ and $g(1w')=g_1(w')$ and $g_1(w')\scoh h_1(w')$. 
\item for all $v < 1w'$ one has $h(v)=g(v)$ since $v=0u$ or $v=1u'$ and 
\begin{itemize}
\item if $v=0u$ then $h(v)=h(0u)=h_0(u)=g_0(u)=g(0u)=g(v)$. 
\item if $v=1u'$ then $u'<w'$ and therefore $h(v)=h(1u')=h_1(u')$, \enspace $h_1(u')=g_1(u')$ (because $u'<w'$),\enspace  and $g(v)=g(1u')=g_1(u')$. 
\end{itemize} 
\end{itemize} 
\end{enumerate} 
\end{proof}

\begin{prop}[$A \text { retract of } \flag A$]  Any coherence space $A$
is a linear retract of the coherence space $\flag  A$, that is there exist two linear maps as in: 
$$
\begin{array}{rrcll} 
&& f_A &&\\[-1ex]  
&  &\--\!\!\!\--\!\!\!{\rightspoon} & &\\[-2ex] 
&A && \flag A&\\[-2ex] 
&& {\leftspoon}\!\!\!\--\!\!\!\-- &&\\[-1ex] 
&& t_A && \\ 
\end{array}
$$ 
\noindent with   $\mathrm{t}_{A} \circ \mathrm{f}_{A}=Id_A$  and  $f_{A} \circ t_{A}\subset Id_{\flag A}$. 
\end{prop}

\begin{proof} 
Consider $Tr(f_A)=\{(a, \underline{a}) \ |\  a \in|A|\}\subset |A| \times |\flag  A|$,  where $\underline{a} \in|\flag  A|$ stands for the constant function mapping any element of $\deux^{\omega}$ to $a$. Let us call $Tr(t_A)$ the symmetric of $Tr(f_A)$, that is $Tr(t_A)=\{(\underline{a},a) \ |\  a \in|A|\}\subset |\flag A| \times |A|$. 
The trace $Tr(f_A)$ is the linear trace of $f_A$  from $A$ to $\flag  A$ while 
$Tr(t_A)$ is 
 the  linear trace of $t_A$ from $\flag  A$ to $A$. The compound $f_{A} \circ t_{A}$ is $\{(\underline{a}, \underline{a}) \ |\  a \in|A|\}$ which is a strict subset of $Id_{\flag A},$ while the compound $\mathrm{t}_{A} \circ \mathrm{f}_{A}$ is exactly
$Id_{A}$
\end{proof} 

\begin{prop}[$\flag$ is an endofunctor of Coh] Given $\ell: A \rightarrow B$ \text { defines } $\flag  \ell: \flag  A \rightarrow \flag  B$ by the following trace:
$$
\flag  \ell=\left\{(f, g) \ |\  \forall w \in \deux^{\omega}(f(w), g(w)) \in \ell\right\}
$$
This makes $\flag$ an endofunctor of the category of $\mathbf{Coh}$ of coherence spaces of definition \ref{Coh}. 
\end{prop}

\begin{proof} 
\noindent\textbf{Firstly, let us show that $\flag  \ell$ defines a linear map from $\flag  A$ to $\flag  B$.}  
Let $(f, g),(f', g')$ be in $\flag  \ell$.
\begin{itemize}
\item 
Assume that $f \scoh f'[\flag  A] .$ Thus there exists $w$ such that $f(w) \scoh f'(w)[A]$ and $f(v)=f'(v)$ for all $v< w$.  By definition of $\flag \ell$,  $(f(w), g(w))$ and $(f'(w), g'(w))$ are in $\ell$ which is linear, so we have $g(w) \coh  g'(w)[B]$. Consider $v<w$, we have $f(v)=a=f'(v)$ and since both $(a, g(v))$ and $(a, g'(v))$
are in $\ell$ which is linear we have $g(v) \coh g'(v)[B]$. 
Applying proposition \ref{diffegalsup}, there exists an $u$ such that $g(u) \neq g'(u)$ and $g(t)=g'(t)$ for all $t<u$. We necessarily have $u \leqslant w$ and therefore $g(u) \scoh  g'(u)[B]$. Hence $g \scoh g'[\flag  B]$. 
\item 
Assume $f=f'$.  For all $w$, both $(f(w), g(w))$ and $(f(w), g'(w))$ are in $\ell$ which
is linear. Therefore, for all $w$, one has $g(w) \coh g'(w)[B]$.  Applying proposition \ref{diffegalsup} either $g=g'$ (1) or there exists an $u$ such that $g(u) \neq g'(u)$ and $g(v)=g'(v)$ for all
$v<u$  (2). When (2) holds, we have $g(u) \scoh  g'(u)[B]$ since $g(w) \coh g'(w)[B]$ for all
$w$ and when (1) holds we have $g=g'$.  In both case we have $g \coh g'[\flag  B]$
\end{itemize} 

\noindent\textbf{It is clear that $\flag  Id_{A}=Id_{\flag A}$} 

\noindent\textbf{Let us finally show that $\flag$  commutes with linear composition.}  Let $\ell: A \multimap B$ and
$\ell': B \multimap C,$ be two linear morphisms. 
\begin{itemize} 
\item 
If $(f, h)$ is in $(\flag  \ell') \circ(\flag  \ell)$ then $(f, h)$ is in $\flag (\ell' \circ \ell)$. Indeed there exists a $g$ in $|\flag  B|$ such that $(f, g)$ is in $\flag  \ell$ and $(g, h)$ is in $\flag  \ell'$. Thus, for all $w$ the pair $(f(w), g(w))$ is in $\ell$ and the pair $(g(w), h(w))$ is in $\ell',$ so that $(f(w), h(w))$ is in $\ell' \circ \ell$ for all $w,$ i.e. $(f, h)$ is in $\flag (\ell' \circ \ell)$
\item 
We now assume that $(f, h)$ is in $\flag (\ell' \circ \ell)$ and show that it is in $(\flag  \ell') \circ(\flag  \ell)$ too.
If $(f, h)$ is in $\flag (\ell' \circ \ell)$ then $(f(w), h(w))$ is in $\ell' \circ \ell \text { for all } w, \text { i.e. for all } w$ there exists a non empty set of tokens $G(w)$ in $|B|$  
with at least one $g(w)\in G(w)$ such that $(f(w), g(w))$ is in $\ell$ and $((g(w), h(w)) \text { is in } \ell' .$ The  point is to show that one may choose, for each $w$ a token $g(w)$ in $G(w)$ in such a way that $g$ is a \emph{continuous function} from $\deux^\omega$ to $|B|$ that is a token in  $|\flag  B|$. Consider the function $(f, h)$ from $\deux^{\omega}$ to $|A| \times|C|$ with the discrete topology
on $|A| \times|C|$. It is continuous, because the product topology of a finite product of discrete spaces is the discrete topology on the (finite) product of the involved sets. Therefore it may be described as a finite binary tree, with leaves in $|A| \times|C| .$ We write it as a finite set $\left\{(m_{i},(a_{i}, c_{i}))\right\}$ with the properties of generic trees (\ref{existecouple}) and
(\ref{unicouple}) given in proposition \ref{repfctcont}. 

For all $i$,  there exists $w$ such that $f(w)=a_{i}$ and $h(w)=c_{i}$ take e.g. $w=m_{i}(0)^{\omega}$.  Therefore for all $i$ there exists $b_{i}$ in $|B|$ such that $(a_{i}, b_{i})$ is in $\ell$ and $(b_{i}, c_{i})$ is in $\ell'$,  and for each $i$ we choose one (there are finitely many $i$). Now, we can define $g(w)$ with the help of the property (\ref{existecouple}) of proposition \ref{repfctcont}.  For each $w$ there exists a unique $i$ such that $w=m_{i} w',$ and we define $g(w)$ to be $b_{i}$,  and thus $g$ is clearly a continuous function from $\deux^{\omega}$ to $|B|$. Notice that the generic tree of $g$ is not necessarily $\left\{(m_{i}, b_{i})\right\}$ but \emph{its normal form} according to \ref{termesnormaux} of  proposition \ref{repfctcont}. Indeed,  the property \ref{unicouple} may fail since there possibly exist $i, j$ and $m$ such that $m_{i}=m 0, m_{j}=m 1$ while $b_{i}=b_{j}$.  Now it is easily seen that $(f, g)$ is in $\flag  \ell$ and $(g, h)$ in $\flag  \ell'$. Indeed, for all $w$ there exists a unique $i$ such that $w=m_{i} w'$ and we then have $f(w)=a_{i}, g(w)=b_{i}$ $h(w)=c_{i}$ and thus $(f(w), g(w))=(a_{i}, b_{i})$ is in $\ell$ and $(g(w), h(w))=(b_{i}, c_{i})$
is in $\ell'$
\end{itemize} 
\end{proof} 

As intuition suggests, $\flag$ is neither a monad nor a comonad. Because $\flag$ is self-dual, $\flag$ is a monad if and only if it is a comonad. The following proof adapted  from \cite{Ret94caen} shows that there is no natural transformation from $\flag$ to the identity endofunctor, 
so $\flag$ is not a comonad. 

\begin{prop} 
The modality $\flag$ is neither a monad nor a comonad. 
\end{prop}

\begin{proof} 
 If $\flag$ were a comonad, then there would exist a natural transformation $\rr$ 
 from the functor $\flag$ to the identity functor $\id$. 
 Let $\un$ be the coherence space whose web is $\{*\}$ 
 and let $\unplusun$ be the coherence space whose web is $\{a,b\}$ with coherence $a\sincoh b$; 
 let $\ell_{a}, \ell_{b}, \ell_{a b}: \unplusun  \multimap \un$ 
 be the three linear morphisms from $\un \oplus \un$ to $\un$ respectively defined by 
$\ell_{a}=\{(a, *)\}$,\enspace $\ell_{b}=\{(b, *)\}$ \enspace $\ell_{a b}=\{(a, *),(b, *)\}$. 

There should exists $\rr_{\unplusun}$ from $\flag(\unplusun)$ to $\unplusun$ and $\rr_\un$ from $\flag(\un)\equiv \un$ to $\un$ such that, for any $x$ 
among $a$, $b$, $ab$:
\begin{center} 
\comsq{x}:\quad $\rr_\un\circ\flag(\ell_x)=\ell_x\circ\rr_\unplusun$ 
\end{center} 

For $x=ab$ one should have $\rr_\un\circ\flag(\ell_{a b})=\ell_{a b}\circ\rr_\unplusun$ 
so it is mandatory that $\rr_{un}\neq \emptyset$  so  $\rr_\un=Id_\un$. Indeed $\rr_\un$ can only be $Id_\un$ or $\emptyset$, because there is no other linear map from $\un$ to itself. 

Now consider the three following elements in $\flag A$ (viewed as terms): 
\begin{verse} 
 $\underline{a}$ (constant function from 
$\deux^\omega$ to $a\in |\unplusun|$),  

$\underline{b}$ (constant function from 
$\deux^\omega$ to $b\in |\unplusun|$), 

$\langle\underline{a},\underline{b}\rangle$ (function from 
$\deux^\omega$ mapping all $0w$ to $a\in |\unplusun|$ and all $1w$ to $b\in |\unplusun|$). 
\end{verse} 

Observe that: 
\begin{verse} 
\comsq{a}  imposes that $\rr_\unplusun$ maps  $\underline{a}$ to $a\in |\unplusun|$, 

\comsq{b}  imposes that $\rr_\unplusun$ maps  $\underline{b}$ to $b\in |\unplusun|$, 
\end{verse} 
so $\langle \underline{a}, \underline{b}\rangle$ cannot be mapped by $\rr_\unplusun$,
hence \comsq{ab}  cannot hold, contradiction. 
\end{proof} 

To close the discussion on  $\flag$ being or not a (co)monad, it is not difficult to be convinced, see  \cite{Ret94caen} for an argument, that there also is no co-associative natural transformation from $\flag$ to $\flag\flag$  --- although it is probably tedious to prove thoroughly

\section{The Modality ``Flag" in the Category of Hypercoherences} 

The construction is very similar, we shall be brief, and closely refer to \cite{Ehr93}. Although we shall be brief, we think it is a good idea to recall what a hypercoherence is. The cardinal of a set $M$ is denoted by $\#M$. 

\begin{dfn}[Reminder on hypercoherences] 
A \emph{hypercoherence} $X$  is defined by its web $|X|$ (a set of tokens) and $\Gamma(X)\subset \powerset^*_{\rm fin}(X)$ i.e. a set of  finite non empty subsets of the web (atomic coherence), which includes all singletons --- strict atomic coherence $\Gamma^*(X)$ is $\Gamma(X)$ minus all singletons.\footnote{
As opposed to Girard's qualitative domains (which include coherence spaces which are the binary generated qualitative domains)\cite{Gir86}, a subpart of an atomic  coherence is not asked to be itself an atomic coherence.}  
The negation $X^\perp$ of an hypercoherence $X=(|X|,\Gamma(X))$ is $X^\perp=(|X|,\powerset^*_{\rm fin}(X)\setminus \Gamma^*(X))$. The tensor product $X \otimes Y$ of two hypercoherences  $X=(|X|,\Gamma(X))$ and  $Y=(|Y|,\Gamma(Y))$
is the hypercoherence $X\times Y=(|X| \times|Y|,  \Gamma(X\times Y))$ with  $w \in \Gamma(X \otimes Y)$ if and only if $w_{1} \in \Gamma(X)$ and $w_{2} \in \Gamma(Y)$ --- $w_1$ and $w_2$ are respectively the projections of $w$ on $|X|$ and on $|Y|$ beware that in general $w\neq (w_1 \times w_2)$.

Given two hypercoherences $X=(|X|,\Gamma(X))$ and  $Y=(|Y|,\Gamma(Y))$ the hypercoherence $X \rightspoon Y$ is defined by $X \rightspoon Y= (|X|\times |Y|, \Gamma(X\rightspoon Y))$ with 
$w\in \Gamma(X\rightspoon Y)$ if and only if  ($w$ is a finite non empty subset of $|X|\times |Y|$ and) the projections $w_1$ and $w_2$ of $w$ on respectively  $|X|$ and on $|Y|$,  satisfy the following implications: 
\begin{itemize} 
\item 
if $w_1\in \Gamma(X)$ then  $w_2\in \Gamma(Y)$ 
\item 
if $w_1\in \Gamma(X)$ and $\#w_1\geq 2$ then $w_2\in \Gamma(Y)$ and $\#w_2\geq 2$ 
\end{itemize} 

A hypercoherent subset of $|X\rightspoon Y|$ can be viewed as a function from $X$ to $Y$.
\footnote{It is too complicated to be presented here.} 
\end{dfn} 

In order to define our flag modality in the category of hypercoherences,  we first have to define the hypercoherence $X \bef Y$ corresponding to the ``before" connective from hypercoherences $|X|$ and $|Y|$: 

\begin{dfn} Let $X$ and $Y$ be hypercoherences. The hypercoherence $X\bef Y$ is the hypercoherence whose web is $|X| \times|Y|$ and whose strict coherence is defined by:
\begin{center} 
$w \in \Gamma^{*}(X)$  \emph{iff}   
$\Big\{ \pi_{1}(w) \in \Gamma^{*}(X)   \mathrm{\ or\ }  \big[ \# \pi_{1}(w)=1 \mathrm{\ and\ } \pi_{2}(w) \in \Gamma^{*}(Y) \big] \Big\}$ 
\end{center} 
\end{dfn} 
It is easily seen that this connective is associative, self-dual, non-commutative and in between the tensor product and the par $-$ just like in the category of coherence spaces.

Now, the self-dual modality enjoying the wanted properties is defined in the category of hypercoherences by:

\begin{dfn}  Let $X$ be an hypercoherence. The hypercoherence $\flag  X$ is the hypercoherence whose web is, as for coherence spaces:
$
|\flag  X|=\gt_{|X|}$
and whose strict atomic coherence is defined by:
$$
\left\{f_{1}, \ldots, f_{k}\right\} \in \Gamma^{*}(X) . \text { iff } . \exists m \in \deux^{\omega}\left\{\begin{array}{l}
\left\{f_{1}(m), \ldots, f_{k}(m)\right\} \in \Gamma^{*}(X) \\
\quad \text { and } \\
\forall m'<m\quad  \#\left\{f_{1}(m'), \ldots, f_{k}(m')\right\}=1
\end{array}\right.
$$
\end{dfn} 

The proofs that the hypercoherence version of $\flag$  enjoys the same properties as the coherence 
version studied in the previous section are the same \textsl{mutatis mutandis}.

\begin{prop} In the category of hypercoherences and linear maps, the modality flag as defined above, enjoys the following properties: 
\begin{itemize}
\item flag is self-dual $(\flag A)^\perp=\flag (A^\perp)$ 
\item $A$ is a retract of $(\flag A)$ 
\item $(\flag A)$  is isomorphic to  $(\flag A)\bef (\flag A)$
\end{itemize}  
\end{prop} 

\section{Hints towards a Syntax for the ``Flag" Modality}

One can view the conclusion of pomset proof as a series parallel partial order of formulas (a partially ordered multiset of formulas, hence the name \emph{pomset logic}). 
If the order 
is viewed as a set of temporal constraints $A\bef B$ meaning the resource $A$  should be consumed before the resource $B$ is.  A modal formula $\flag A$ can be understood as repeatedly $A$, or $A$ at any instant --- while $!A$ rather means as many $A$ as you wish at the instant where it is located. 

A part from proof-nets, deductive  systems for pomset logic are quite difficult to present, and since we do not yet have a satisfying syntax for the modality flag, we think it is wiser not to impose to the reader a technical section that does not yet  yield a precise result. We nevertheless want to informally discuss what rules for the flag modality may look like. 

The easiest proof system for pomset logic are proof-nets (cf. subsection \ref{cohIntPN}). Notice that so far the only inductive definition of those proof-nets is the very complicated sequent calculus recently provided by Slavnov \cite{Slavnov2019}, where sequents $\seq A_1,\ldots,A_n$ are endowed with a family of binary relations $R_k$ (with $1\leq k\leq n$) between pairs of tuples of $k$ conclusions $(A_{i_1},\ldots,A_{i_k})$ without common elements among 
$A_{i_1},\ldots,A_{i_k}$ --- the relation $R_k$should stable when one applies the same permutation of the $k$ indices to the two tuples of lenght $k$  --- corresponding to disjoint paths between $k$ conclusions in a proof-net. 

The calculus of structures SBV \cite{GugStr01} is a subsystem of Pomset Logic \cite{Retore2021Lambek}  where the connective before is called ``seq(ential)" and SBV is  a strict subsystem of pomset logic, according to some recent  work by Nguyen and Stra{\ss}burger, yet unpublished \cite{NguyenStrass2021}. This calculus consists in rewriting formulas (as terms) up to associativity of all the three connectives $\pa,\bef,\ts$ and the commutativity of the connective $\pa,\ts$, starting with the axiom $\ts_i (a_i\pa a_i^\perp)$ --- some $a_i$ may denote the same propositional variables. All the tautologies of SBV correspond to pomset proof-nets. 

The contraction morphism for flag $(\flag A \bef \flag A) \rightspoon \flag A$ is easily implemented with proof-nets.  The initial two conclusions $\flag A$ and $\flag A$  are at the same place in the order (at the same instant), and  a ``before" link between those two conclusions as premises yields $\flag A$ instead of $\flag A \bef \flag A$.  The correctness criterion applies. This is easily implemented in SBV: $T[(\flag A \bef \flag A)]$  rewrites to $T[\flag A]$.  This quite reminiscent of the rule for contraction of $?A$ in linear logic, with $\pa$ being replaced with $\bef$.

The duplication is easy to  import in a  system like SBV:  $\flag A\rightspoon (\flag A \bef \flag A) $ (duplication) is interpreted as a term rewrite rule. Thus a term $T[(\flag A \bef \flag A)]$  rewrites to $T[\flag A]$ and as usual in SBV this rewriting may take place within a formula. However handsome proof-nets do not handle multiple conclusions that way and consequently, without any additional structure, the rule cannot be modelled that way. What we learnt  from $?$ and $!$ is that such a rule is likely to be unnecessary, because $\flag A$ is the dual of $\flag A$, so reducing a cut between two contraction suggests what the rule may look like. 

The introduction yielding  $\flag A$ from $A$ is much more problematic. For  $!A$ the context ought to be duplicable i.e. to be $?\Gamma$. Here the context ought to be replicable in time hence of the shape $\flag \Gamma$, because the dual of  $\flag X$ is $\flag X$. This would require to consider $\flag$ boxes; boxes are not fully satisfactory, but may work.  When conclusions are  $\flag \Gamma$ and $A$ one could think that a box with conclusions $\flag \Gamma, \flag A$ could be introduced, but it is not that simple: an axiom $a,\lnot a$ cannot be turned into $\flag a, \flag a^\perp$, because as opposed to $!$ whose dual is $?$ that can be freely introduced by the dereliction rule, the dual of $\flag$
 is the very same modality $\flag$.  A possible solution is to make a box and to simultaneously introduce $\flag$ on all the conclusions of the box. 

An important guideline is to study how the rules that we want to design would preserve the cut-elimination property. A cut between $\flag A$ and $\flag A^\perp$ should reduce to cut(s) between $A$ and $A^\perp$. When both $\flag$ result from a box (which can be the same bix, because of self-duality) ,  or when they both come from a contraction (as in Gentzen MIX rule also known as cross-cuts), we do not yet know how those cuts should be reduced. 
 
The proposal in \cite{gug2017lyon} motivated by a computational analysis of cut elimination  rather use rules like $((\flag A)\pa A)\rightspoon  \flag A$. Such rules can be viewed as more specific, or particular cases of the ones we just discussed  because  $(A\bef B)\rightspoon (A\pa B)$.

So we think it is too early to say more than the above remarks about the possible syntax of $\flag$. 

\section{Conclusion} 

A challenging issue is to define a deductive system enjoying cut elimination including a syntactical 
match of 
the self dual modality presented in this paper, that is a calculus  
for non-commutative  contraction and duplication.  This logical calculus could be defined as an extension of the calculus of structures with deep inference (roughly speaking, internal rewriting)
\cite{Gug99,GugStr01,Gug2007,GugStrass2011tcl}
or with pomset proof-nets with or without links \cite{Ret97tlca,Ret98roma,Retore2021Lambek}, 
or with one of the sequent calculi introduced for pomset logic \cite{Ret93,Ret97tlca,Retore2021Lambek} not to forget the complex proposal by Sergey Slavnov 
\cite{Slavnov2019} which is complete with respect to  pomset proof-nets. 
The ongoing syntactic work by Alessio Guglielmi with a self dual modality \cite{gug2017lyon} looks quite appealing. 

Rules may emerge from  our intuition in terms of models of computation, a standard applicative domain for linear logic and related formalisms,.  that has been well developed for the calculus of structure and deep inference.

We think that this ancient work may give guidelines for defining such a deductive system as coherence semantics often did for linear logic.   Such a logic would be a multiplicative exponential non-commutative  linear logic close to classical logic, as MELL is to intuitionistic logic.

\paragraph*{Acknowledgements:}  Thanks to Pierre Ageron, Thomas Ehrhard, Catherine Gourion,  Thomas Streicher and especially to Jean-Yves Girard, Achim Jung  and Myriam Quatrini for their constructive comments. Thanks to the reviewers of both the conference paper and the proceedings, their suggestions greatly contributed to improve the conference version as well as the proceedings version. 

\bibliography{entcsLinearityRetore} 

\begin{thebibliography}{10}
\providecommand{\bibitemdeclare}[2]{}
\providecommand{\surnamestart}{}
\providecommand{\surnameend}{}
\providecommand{\urlprefix}{Available at }
\providecommand{\url}[1]{\texttt{#1}}
\providecommand{\href}[2]{\texttt{#2}}
\providecommand{\urlalt}[2]{\href{#1}{#2}}
\providecommand{\doi}[1]{doi:\urlalt{http://dx.doi.org/#1}{#1}}
\providecommand{\bibinfo}[2]{#2}

\bibitemdeclare{inproceedings}{BGR97}
\bibitem{BGR97}
\bibinfo{author}{Denis \surnamestart Bechet\surnameend},
  \bibinfo{author}{Philippe \surnamestart de~Groote\surnameend} \&
  \bibinfo{author}{Christian \surnamestart Retor\'e\surnameend}
  (\bibinfo{year}{1997}): \emph{\bibinfo{title}{A complete axiomatisation of
  the inclusion of series-parallel partial orders}}.
\newblock In \bibinfo{editor}{H.~\surnamestart Comon\surnameend}, editor: {\sl
  \bibinfo{booktitle}{Rewriting Techniques and Applications, RTA`97}}, {\sl
  \bibinfo{series}{LNCS}} \bibinfo{volume}{1232}, \bibinfo{publisher}{Springer
  Verlag}, pp. \bibinfo{pages}{230--240}, \doi{10.1007/3-540-62950-5_74}.

\bibitemdeclare{article}{Ehr93}
\bibitem{Ehr93}
\bibinfo{author}{Thomas \surnamestart Ehrhard\surnameend}
  (\bibinfo{year}{1993}): \emph{\bibinfo{title}{Hypercoherences: a strongly
  stable model of linear logic}}.
\newblock {\sl \bibinfo{journal}{Mathematical {S}tructures in {C}omputer
  {S}cience}} \bibinfo{volume}{3}(\bibinfo{number}{4}), pp.
  \bibinfo{pages}{365--385}, \doi{10.1017/S0960129500000281}.

\bibitemdeclare{article}{Gen34}
\bibitem{Gen34}
\bibinfo{author}{Gehrard \surnamestart Gentzen\surnameend}
  (\bibinfo{year}{1934}): \emph{\bibinfo{title}{Untersuchungen {\"u}ber das
  logische {S}chlie\ss en {I}}}.
\newblock {\sl \bibinfo{journal}{Mathematische {Z}eitschrift}}
  \bibinfo{volume}{39}, pp. \bibinfo{pages}{176--210},
  \doi{10.1007/BF01201353}.
\newblock \bibinfo{note}{Traduction Fran\c caise de R.~Feys et J.~Ladri\`ere:
  Recherches sur la d\'eduction logique, Presses Universitaires de France,
  Paris, 1955}.

\bibitemdeclare{article}{Gir86}
\bibitem{Gir86}
\bibinfo{author}{Jean-Yves \surnamestart Girard\surnameend}
  (\bibinfo{year}{1986}): \emph{\bibinfo{title}{The System {F} of Variable
  Types: Fifteen Years Later}}.
\newblock {\sl \bibinfo{journal}{Theoretical Computer Science}}
  \bibinfo{volume}{45}, pp. \bibinfo{pages}{159--192},
  \doi{10.1016/0304-3975(86)90044-7}.

\bibitemdeclare{article}{Gir87}
\bibitem{Gir87}
\bibinfo{author}{Jean-Yves \surnamestart Girard\surnameend}
  (\bibinfo{year}{1987}): \emph{\bibinfo{title}{Linear Logic}}.
\newblock {\sl \bibinfo{journal}{Theoretical Computer Science}}
  \bibinfo{volume}{50}(\bibinfo{number}{1}), pp. \bibinfo{pages}{1--102},
  \doi{10.1016/0304-3975(87)90045-4}.

\bibitemdeclare{article}{Gir91}
\bibitem{Gir91}
\bibinfo{author}{Jean-Yves \surnamestart Girard\surnameend}
  (\bibinfo{year}{1991}): \emph{\bibinfo{title}{A new constructive logic:
  classical logic}}.
\newblock {\sl \bibinfo{journal}{Mathematical Structures in Computer Science}}
  \bibinfo{volume}{1}(\bibinfo{number}{3}), pp. \bibinfo{pages}{255--296},
  \doi{10.1017/S0960129500001328}.

\bibitemdeclare{book}{Girard2011blindspot}
\bibitem{Girard2011blindspot}
\bibinfo{author}{Jean-Yves \surnamestart Girard\surnameend}
  (\bibinfo{year}{2011}): \emph{\bibinfo{title}{The blind spot -- lectures on
  logic}}.
\newblock \bibinfo{publisher}{European Mathematical Society},
  \doi{10.4171/088}.

\bibitemdeclare{techreport}{Gug99}
\bibitem{Gug99}
\bibinfo{author}{Alessio \surnamestart Guglielmi\surnameend}
  (\bibinfo{year}{1999}): \emph{\bibinfo{title}{A Calculus of Order and
  Interaction}}.
\newblock \bibinfo{type}{Technical Report} \bibinfo{number}{WV-99-04},
  \bibinfo{institution}{Dresden University of Technology}.

\bibitemdeclare{article}{Gug2007}
\bibitem{Gug2007}
\bibinfo{author}{Alessio \surnamestart Guglielmi\surnameend}
  (\bibinfo{year}{2007}): \emph{\bibinfo{title}{A System of Interaction and
  Structure}}.
\newblock {\sl \bibinfo{journal}{ACM Transactions on Computational Logic}}
  \bibinfo{volume}{8}(\bibinfo{number}{1}), pp. \bibinfo{pages}{1--64},
  \doi{10.1145/1182613.1182614}.

\bibitemdeclare{misc}{gug2017lyon}
\bibitem{gug2017lyon}
\bibinfo{author}{Alessio \surnamestart Guglielmi\surnameend}
  (\bibinfo{year}{2017}): \emph{\bibinfo{title}{Decoupling normalization
  mechanisms with an eye toward concurrency}}.
\newblock \bibinfo{howpublished}{Talk at ENS LYON seminar: programs and
  proofs.}
\newblock \bibinfo{note}{Slides: \url{http://cs.bath.ac.uk/ag/t/DNMWAETC.pdf}}.

\bibitemdeclare{inproceedings}{GugStr01}
\bibitem{GugStr01}
\bibinfo{author}{Alessio \surnamestart Guglielmi\surnameend} \&
  \bibinfo{author}{Lutz \surnamestart Stra{\ss}burger\surnameend}
  (\bibinfo{year}{2001}): \emph{\bibinfo{title}{Non-commutativity and {MELL} in
  the Calculus of Structures}}.
\newblock In \bibinfo{editor}{L.~\surnamestart Fribourg\surnameend}, editor:
  {\sl \bibinfo{booktitle}{CSL 2001}}, {\sl \bibinfo{series}{Lecture Notes in
  Computer Science}} \bibinfo{volume}{2142},
  \bibinfo{publisher}{Springer-Verlag}, pp. \bibinfo{pages}{54--68},
  \doi{10.1007/3-540-44802-0_5}.

\bibitemdeclare{article}{GugStrass2011tcl}
\bibitem{GugStrass2011tcl}
\bibinfo{author}{Alessio \surnamestart Guglielmi\surnameend} \&
  \bibinfo{author}{Lutz \surnamestart {Stra{\ss}burger}\surnameend}
  (\bibinfo{year}{2011}): \emph{\bibinfo{title}{A System of Interaction and
  Structure IV: The Exponentials and Decomposition}}.
\newblock {\sl \bibinfo{journal}{ACM transaction of computational logic}}
  \bibinfo{volume}{12}(\bibinfo{number}{4}), p.~\bibinfo{pages}{23},
  \doi{10.1145/1970398.1970399}.

\bibitemdeclare{incollection}{GuoSurmacs2018dicographs}
\bibitem{GuoSurmacs2018dicographs}
\bibinfo{author}{Yubao \surnamestart Guo\surnameend} \& \bibinfo{author}{Michel
  \surnamestart Surmacs\surnameend} (\bibinfo{year}{2018}):
  \emph{\bibinfo{title}{Miscellaneous Digraph Classes}}.
\newblock In \bibinfo{editor}{J.~\surnamestart Bang-Jensen\surnameend} \&
  \bibinfo{editor}{G.~\surnamestart Gutin\surnameend}, editors: {\sl
  \bibinfo{booktitle}{Classes of Directed Graphs}},
  chapter~\bibinfo{chapter}{11}, \bibinfo{publisher}{Springer}, pp.
  \bibinfo{pages}{517--574}, \doi{10.1007/978-3-319-71840-8_11}.

\bibitemdeclare{inproceedings}{Loader94}
\bibitem{Loader94}
\bibinfo{author}{Ralph \surnamestart Loader\surnameend} (\bibinfo{year}{1994}):
  \emph{\bibinfo{title}{Linear Logic, Totality and Full Completeness}}.
\newblock In: {\sl \bibinfo{booktitle}{LICS: IEEE symposium on Logic In
  Computer Science}}, pp. \bibinfo{pages}{292--298},
  \doi{10.1109/LICS.1994.316060}.

\bibitemdeclare{misc}{nguyen2019proof}
\bibitem{nguyen2019proof}
\bibinfo{author}{L{\^e}~Th{\`a}nh~D{\~u}ng \surnamestart Nguy{\^e}n\surnameend}
  (\bibinfo{year}{2019}): \emph{\bibinfo{title}{Proof nets through the lens of
  graph theory: a compilation of remarks}}.
\newblock \urlprefix\url{https://arxiv.org/abs/1912.10606}.

\bibitemdeclare{unpublished}{NguyenStrass2021}
\bibitem{NguyenStrass2021}
\bibinfo{author}{Tito \surnamestart Nguyen\surnameend} \& \bibinfo{author}{Lutz
  \surnamestart Stra{\ss}burger\surnameend} (\bibinfo{year}{2021}):
  \emph{\bibinfo{title}{A complexity gap between pomset logic and system BV}}.
\newblock \bibinfo{note}{Talk at the informal workshop on Proof-Net of the
  GDR-I Linear Logic}.

\bibitemdeclare{phdthesis}{Qua95}
\bibitem{Qua95}
\bibinfo{author}{Myriam \surnamestart Quatrini\surnameend}
  (\bibinfo{year}{1995}): \emph{\bibinfo{title}{S\'emantique coh\'erente des
  exponentielles: de la logique lin\'eaire \`a la logique classique}}.
\newblock \bibinfo{type}{Th\`ese de {D}octorat, sp\'ecialit\'e
  {M}ath\'ematiques}, \bibinfo{school}{Universit\'e Aix-Marseille 2}.

\bibitemdeclare{phdthesis}{Ret93}
\bibitem{Ret93}
\bibinfo{author}{Christian \surnamestart Retor\'e\surnameend}
  (\bibinfo{year}{1993}): \emph{\bibinfo{title}{R\'eseaux et S\'equents
  Ordonn\'es}}.
\newblock \bibinfo{type}{Th\`ese de {D}octorat, sp\'ecialit\'e
  {M}ath\'ematiques}, \bibinfo{school}{Universit\'e Paris 7}.
\newblock \urlprefix\url{https://tel.archives-ouvertes.fr/tel-00585634}.

\bibitemdeclare{techreport}{Ret94rc}
\bibitem{Ret94rc}
\bibinfo{author}{Christian \surnamestart Retor\'e\surnameend}
  (\bibinfo{year}{1994}): \emph{\bibinfo{title}{On the relation between
  coherence semantics and multiplicative proof nets}}.
\newblock \bibinfo{type}{Rapport de {R}echerche} \bibinfo{number}{RR-2430},
  \bibinfo{institution}{INRIA}.
\newblock \urlprefix\url{https://hal.inria.fr/inria-00074245}.

\bibitemdeclare{techreport}{Ret94mod}
\bibitem{Ret94mod}
\bibinfo{author}{Christian \surnamestart Retor\'e\surnameend}
  (\bibinfo{year}{1994}): \emph{\bibinfo{title}{A self-dual modality for
  "before" in the category of coherence spaces and in the category of
  hypercoherences.}}
\newblock \bibinfo{type}{Rapport de {R}echerche} \bibinfo{number}{RR-2432},
  \bibinfo{institution}{INRIA}.
\newblock \urlprefix\url{https://hal.inria.fr/inria-00074243}.

\bibitemdeclare{inproceedings}{Ret94caen}
\bibitem{Ret94caen}
\bibinfo{author}{Christian \surnamestart Retor\'e\surnameend}
  (\bibinfo{year}{1994}): \emph{\bibinfo{title}{Une modalit\'e autoduale pour
  le connecteur ``pr\'ec\`ede''}}.
\newblock In \bibinfo{editor}{Pierre \surnamestart Ageron\surnameend}, editor:
  {\sl \bibinfo{booktitle}{Cat\'egories, Alg\`ebres, Esquisses et
  N\'eo-Esquisses}}, \bibinfo{series}{Publications du D\'epartement de
  Math\'ematiques, Universit\'e de Caen}, pp. \bibinfo{pages}{11--16}.

\bibitemdeclare{inproceedings}{Ret96entcs}
\bibitem{Ret96entcs}
\bibinfo{author}{Christian \surnamestart Retor\'{e}\surnameend}
  (\bibinfo{year}{1996}): \emph{\bibinfo{title}{Perfect matchings and
  series-parallel graphs: multiplicative proof nets as {R}\&{B}-graphs}}.
\newblock In \bibinfo{editor}{J.-Y. \surnamestart Girard\surnameend},
  \bibinfo{editor}{M.~\surnamestart Okada\surnameend} \&
  \bibinfo{editor}{A.~\surnamestart Scedrov\surnameend}, editors: {\sl
  \bibinfo{booktitle}{Linear`96}}, {\sl \bibinfo{series}{Electronic Notes in
  Theoretical Science}}~\bibinfo{volume}{3}, \bibinfo{publisher}{Elsevier}, pp.
  \bibinfo{pages}{167--182}, \doi{10.1016/S1571-0661(05)80416-5}.
\newblock \bibinfo{note}{Http://www.elsevier.nl/}.

\bibitemdeclare{inproceedings}{Ret97tlca}
\bibitem{Ret97tlca}
\bibinfo{author}{Christian \surnamestart Retor\'e\surnameend}
  (\bibinfo{year}{1997}): \emph{\bibinfo{title}{Pomset logic: a non-commutative
  extension of classical linear logic}}.
\newblock In \bibinfo{editor}{Philippe \surnamestart de~Groote\surnameend} \&
  \bibinfo{editor}{James~Roger \surnamestart Hindley\surnameend}, editors: {\sl
  \bibinfo{booktitle}{Typed Lambda Calculus and Applications, TLCA'97}}, {\sl
  \bibinfo{series}{LNCS}} \bibinfo{volume}{1210}, pp.
  \bibinfo{pages}{300--318}, \doi{10.1007/3-540-62688-3_43}.

\bibitemdeclare{article}{Ret97}
\bibitem{Ret97}
\bibinfo{author}{Christian \surnamestart Retor\'e\surnameend}
  (\bibinfo{year}{1997}): \emph{\bibinfo{title}{A semantic characterisation of
  the correctness of a proof net}}.
\newblock {\sl \bibinfo{journal}{Mathematical Structures in Computer Science}}
  \bibinfo{volume}{7}(\bibinfo{number}{5}), pp. \bibinfo{pages}{445--452},
  \doi{10.1017/S096012959700234X}.

\bibitemdeclare{incollection}{Ret98roma}
\bibitem{Ret98roma}
\bibinfo{author}{Christian \surnamestart Retor\'e\surnameend}
  (\bibinfo{year}{1999}): \emph{\bibinfo{title}{Pomset logic as a calculus of
  directed cographs}}.
\newblock In \bibinfo{editor}{V.~M. \surnamestart Abrusci\surnameend} \&
  \bibinfo{editor}{C.~\surnamestart Casadio\surnameend}, editors: {\sl
  \bibinfo{booktitle}{Dynamic Perspectives in Logic and Linguistics: Proof
  Theoretical Dimensions of Communication Processes,Proceedings of the 4th Roma
  Workshop}}, \bibinfo{publisher}{Bulzoni}, \bibinfo{address}{Roma}, pp.
  \bibinfo{pages}{221--247}.
\newblock \bibinfo{note}{INRIA Research Report RR-3714
  \url{https://hal.inria.fr/inria-00072953}}.

\bibitemdeclare{incollection}{Retore2021Lambek}
\bibitem{Retore2021Lambek}
\bibinfo{author}{Christian \surnamestart Retor\'e\surnameend}
  (\bibinfo{year}{2021}): \emph{\bibinfo{title}{Pomset logic: The other
  approach to non commutativity in logic}}.
\newblock In \bibinfo{editor}{Claudia \surnamestart Casadio\surnameend} \&
  \bibinfo{editor}{Philip~J. \surnamestart Scott\surnameend}, editors: {\sl
  \bibinfo{booktitle}{Joachim Lambek: on the interplay of mathematics, logic
  and linguistics}}, \bibinfo{series}{Outstanding contributions to logic},
  \bibinfo{publisher}{Springer Verlag}, pp. \bibinfo{pages}{299--246},
  \doi{10.1007/978-3-030-66545-6}.

\bibitemdeclare{article}{Slavnov2019}
\bibitem{Slavnov2019}
\bibinfo{author}{Sergey \surnamestart Slavnov\surnameend}
  (\bibinfo{year}{2019}): \emph{\bibinfo{title}{{On noncommutative extensions
  of linear logic}}}.
\newblock {\sl \bibinfo{journal}{{Logical Methods in Computer Science}}}
  \bibinfo{volume}{{Volume 15, Issue 3}}, \doi{10.23638/LMCS-15(3:30)2019}.
\newblock \urlprefix\url{https://lmcs.episciences.org/5774}.

\bibitemdeclare{book}{Tro92}
\bibitem{Tro92}
\bibinfo{author}{Anne~Sjerp \surnamestart Troelstra\surnameend}
  (\bibinfo{year}{1992}): \emph{\bibinfo{title}{Lectures on Linear Logic}}.
\newblock {\sl \bibinfo{series}{CSLI Lecture Notes}}~\bibinfo{volume}{29},
  \bibinfo{publisher}{CSLI}.
\newblock \bibinfo{note}{(distributed by Cambridge University Press)}.

\end{thebibliography}
\bibliographystyle{eptcs} 

\end{document}